\documentclass[a4paper,UKenglish]{lipics-v2018-arxiv}

\usepackage{microtype}

\usepackage{tikz}

\theoremstyle{plain}
\newtheorem{proposition}[theorem]{Proposition}

\newcommand{\pstart}{\operatorname{start}}
\newcommand{\pend}{\operatorname{end}}
\newcommand{\istart}{\operatorname{l}}
\newcommand{\iend}{\operatorname{r}}
\newcommand{\greedyp}{\mathbf{GP}}

\newcommand{\rev}{\operatorname{rev}}

\title{On Streaming Algorithms for the Steiner Cycle and Path Cover Problem on
Interval Graphs and Falling Platforms in Video Games}

\titlerunning{Steiner Cycle and Path Cover on Interval Graphs and Falling
Platforms}

\author{Ante Ćustić}{Department of Mathematics, Simon Fraser
    University\\{[250-13450 102nd Avenue, Surrey, BC Canada V3T 0A3]}}{acustic@sfu.ca}{0000-0002-4616-2932}{}

\author{Stefan Lendl}{Institute of Discrete Mathematics, Graz University of
    Technology\\{[Steyrergasse 30, 8010 Graz, Austria]}}{lendl@math.tugraz.at}{0000-0002-5660-5397}{}

\authorrunning{A. Ćustić and S. Lendl}

\Copyright{Ante Ćustić and Stefan Lendl}



\category{}

\relatedversion{}

\supplement{}

\funding{The authors acknowledge the support of the Austrian Science Fund (FWF): W1230.}

\acknowledgements{}

\EventEditors{John Q. Open and Joan R. Access}
\EventNoEds{2}
\EventLongTitle{42nd Conference on Very Important Topics (CVIT 2016)}
\EventShortTitle{arXiv}
\EventAcronym{CVIT}
\EventYear{2016}
\EventDate{December 24--27, 2016}
\EventLocation{Little Whinging, United Kingdom}
\EventLogo{}
\SeriesVolume{42}
\ArticleNo{}
\nolinenumbers 
\hideLIPIcs  

\begin{document}

\maketitle

\begin{abstract}
We introduce a simplified model for platform game levels with falling
platforms based on interval graphs and show that solvability of such levels
corresponds to finding Steiner cycles or Steiner paths in the corresponding
graphs. Linear time algorithms are obtained for both of these problems.
We also study these algorithms as streaming algorithms and analyze the necessary
memory with respect to the maximum number of intervals contained in another interval. This
corresponds to understanding which parts of a level have to be visible at each
point to allow the player to make optimal deterministic decisions.
\end{abstract}

\section{Introduction}

In 2D platform games it is a common game mechanism to include platforms that
fall or break
after the player visits them once. Additionally, it is often the case that the
player has to collect certain items (coins, stars, \dots) that are placed on
some of these platforms and afterwards get back to the start or reach the exit of the level.
Popular examples of video games that are (partially) based on these principles include
Super Mario Bros., Donkey Kong Country and Super Mario Land\footnote{Super Mario
    Bros., Donkey Kong Country and Super Mario land are a trademarks of Nintendo. Sprites
are used here under Fair Use for educational purposes.} (see
Figure~\ref{fig:gameex}).
We study solvability of levels based on these  principles by introducing a toy model
of such video games, in which all platforms (except for the target/starting point) have this falling
property. The reachability between two platforms is modeled via an interval
graph, which in many cases is a reasonable simplification.
Then, the solvability of a level boils down to either finding a \emph{Steiner cycle}
or a \emph{Steiner path} in the corresponding interval graph.
To our knowledge, these problems have not been studied for this specific graph
class. The \emph{Hamiltonian cycle} and \emph{Hamiltonian
path} problem, which are special cases of the Steiner variants, are extensively studied for interval graphs and can be solved in linear
time, if the intervals are given as a right endpoint sorted list~\cite{keil1985finding,
arikati1990linear,manacher1990optimum}.

\begin{figure}
	\centering
    \begin{subfigure}[b]{0.46\textwidth}
        \includegraphics[width=\textwidth]{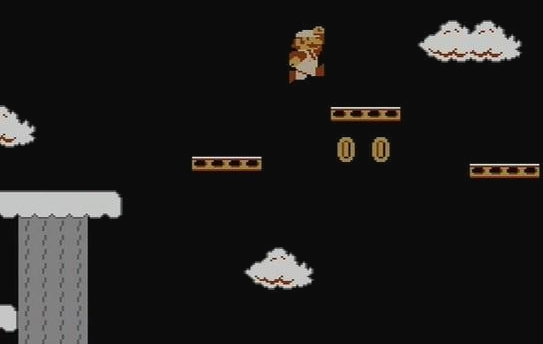}
        \caption{Super Mario Bros. (1985, NES)}
        \label{fig:mario}
    \end{subfigure}
    \qquad
    \begin{subfigure}[b]{0.46\textwidth}
        \includegraphics[width=\textwidth]{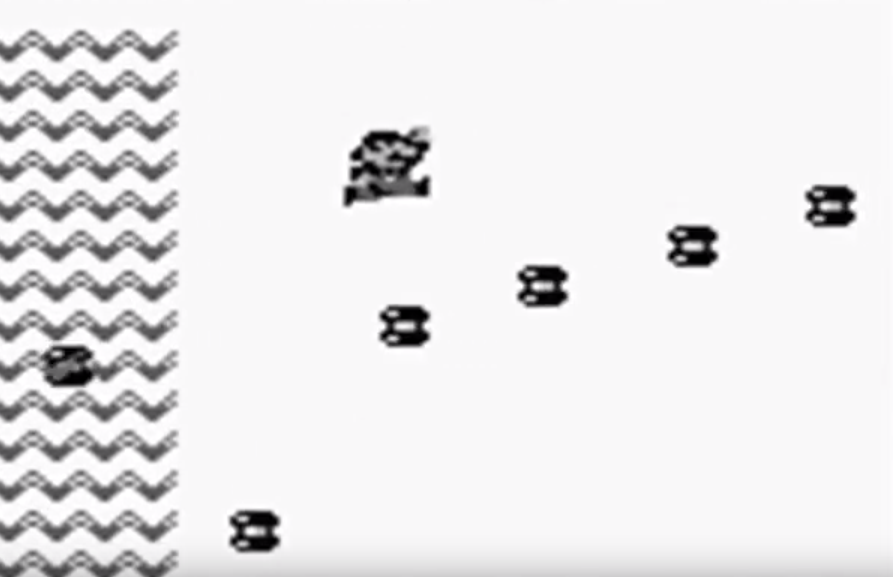}
        \caption{Super Maro Land (1989, Game Boy)}
    \end{subfigure}

    \begin{subfigure}[b]{0.46\textwidth}
        \includegraphics[width=\textwidth]{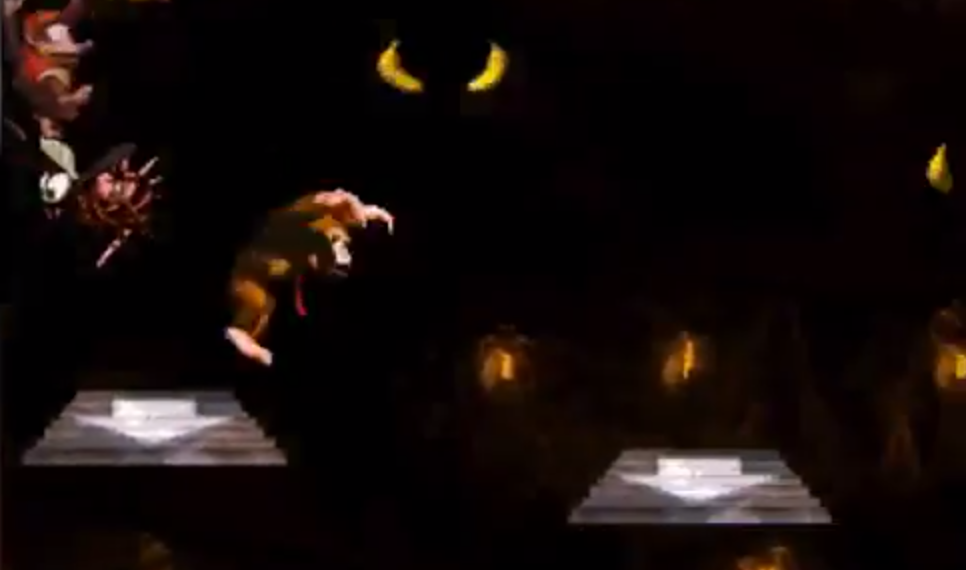}
        \caption{Donkey Kong Country (1994, SNES)}
    \end{subfigure}
    \caption{Examples of levels from popular platform games including falling
    platforms.}
    \label{fig:gameex}
\end{figure}

In this work we generalize the algorithms of
Manacher~et~al.~\cite{manacher1990optimum} to the \emph{Steiner setting} and
obtain first linear time algorithms for the \emph{Steiner path cover} and
\emph{Steiner cycle} problem on interval graphs.
A second important aspect when considering 2D game levels is the fact that the
screen size is limited, so the whole level is not visible to the player at once.
By studying our algorithms as \emph{single pass streaming algorithms} we 
state precisely which parts of a level have to be visible to the player to
deterministically decide how to play at each time. Alternatively, this
can be interpreted as a memory bound for the streaming algorithms in terms of a
natural graph parameter for interval graphs.

For a more general model for platform game levels based on intersection graphs of
two dimensional boxes these problems are known to be NP-hard. Such graphs are
generalizations of grid graphs for which already the Hamiltonian path problem is
known to be NP-hard~\cite{itai1982hamilton}.

\section{Definitions and Preliminary Results}

Given an interval $i = [x, y]$ we denote the starting point $x$ by $\istart(i) = x$ and
the endpoint $y$ by $\iend(i) = y$. Let $I = (i_{1}, i_{2}, \dots, i_{n})$ be
a list or set of intervals. We denote by $G(I)$ the interval graph of $I$. The
vertices of this graph correspond to the intervals of $I$. Two intervals $i, i'
\in I$ are connected by an edge in $G(I)$ if $i \cap i' \neq \emptyset$.

For an arbitrary graph $G = (V, E)$ a list of vertices $P = (i_{1}, i_{2},
\dots, i_{l})$ is a (simple) path if those vertices are pairwise distinct and for
each $j=1,2,\dots,l-1$ it holds that $\{i_{j}, i_{j+1}\} \in E$. The start of
$P$ is denoted by $\pstart(P) = i_{1}$ and the end of $P$ is denoted by
$\pend(P) = i_{l}$. We define $\rev(P)$ as the reverse path $(i_{l}, i_{l-1},
\dots, i_{1})$ of $P$. If in addition
$\{i_{l}, i_{1}\} \in E$ we call $P$ a (simple) cycle. For ease of writing we
sometimes abuse notation and consider $P$ as a set instead of a list, to allow
for the use of set operations. Given two paths $P$ and $Q$ and a vertex $i$ we also write $(P,Q)$
for the concatenation of $P$ and $Q$ and $(P,i)$ for the concatenation of $P$
and $i$.
Given a set $S \subseteq V$, a \emph{Steiner cycle} is a
cycle $C$ in $G$ such that $S \subseteq C$. A \emph{Steiner path cover}
of $G$ is a set $\{P_{1}, P_{2}, \dots, P_{k}\}$ of paths in $G$ such that $S
\subseteq \bigcup_{j=1}^{k} P_{j}$. The \emph{Steiner path cover number} $\pi_{S}(G)$ is the
the minimum cardinality of a Steiner path cover. If $\pi_{S}(G) = 1$ we say that
$G$ has a \emph{Steiner path}. A set $C \subseteq V$ is called a \emph{cutset}
of $G$ if $G-C$ is disconnected. A set of vertices $T \subseteq V$ is called an
island with respect to $C$, if $T$ is not adjacent to any vertex in $V\setminus(C
\cup T)$. $T$ is called an $S$-island with respect to $C$, if $T$ is an island
with respect to $C$ and $S \cap T \neq \emptyset$.

The following two results are generalizations of two  observations by
Hung~and~Chung~\cite{hung2011linear}, easily verified by the pigeonhole principle.

\begin{proposition}\label{proposition:pathcover}
    Let $C$ be a cutset of $G$ and $g_{S}$ the number of connected components $K$ in
    $G-C$ such that $K \cap S \neq \emptyset$. Then, $\pi_{S}(G) \geq g_{S} -
    |C|$.
\end{proposition}

\begin{proposition}\label{proposition:steinercycle}
    Let $C$ be a cutset of $G$ and $g_{S}$ the number of connected components $K$ in
    $G-C$ such that $K \cap S \neq \emptyset$. If $g_{S} > |C|$, then $G$ has no
    Steiner cycle.
\end{proposition}

 We use these results 
to solve the \emph{Steiner path cover} problem
(see Section~\ref{sec:algpathcover}) and the \emph{Steiner cycle} problem
(see Section~\ref{sec:algsteinercyc}) on interval graphs efficiently.
In the whole paper we assume that $|S|$ is known to the algorithms and queries $i
\in S$ can be performed in $O(1)$ time.

\section{The Steiner Path Cover Problem}\label{sec:algpathcover}

We show that the basic greedy principle, that is the core of efficient
algorithms for the \emph{path cover} problem on interval graphs,
can be generalized by the introduction of neglectable intervals.
The basic greedy principle to find paths in interval graphs was introduced independently by
Manacher~et.~al~\cite{manacher1990optimum} and
Arikati~et~al.~\cite{arikati1990linear}.

Given a right endpoint sorted list of
interval $i_{1}, i_{2}, \dots, i_{n}$ the algorithm iteratively constructs a
path $P$. It starts with the path $P := (i_{1})$ containing only the first
interval. Then it repetitively extends $P$ by
the neighbor of $\pend(P)$ not in $P$ with minimum right endpoint. If no such
extension is possible the algorithm terminates with the current path $P$ as an
output. We denote this algorithm by $\greedyp$ and the path $P$ obtained by this
algorithm by $\greedyp(I)$.

For a path $P = \greedyp(I) = (i_{1}, i_{2}, \dots, i_{l})$ obtained by the
algorithm if executed on an interval graph $G(I)$, we define 
$L(P)$, the set of intervals that exceed beyond the right endpoint of the end of
$P$, i.e. $L(P) = \{ i \in P \colon \iend(i) > \iend(\pend(P)) \}$. Based on
this we recursively define $C(P)$, the set of covers of the path $P$. If $L(P) =
\emptyset$, we also set $C(P) = \emptyset$. Otherwise, let $j$ be the maximum
index such that $i_{j} \in L(P)$. We set $C(P) = \{i_{j}\} \cup C(P')$ for $P' =
(i_{1}, i_{2}, \dots, i_{j-1})$.

For  $C(P) = \{c_{1}, c_{2}, \dots c_{k}
\}$ and $P = (P_{0}, c_{1}, P_{1}, c_{2}, \dots, c_{k}, P_{k})$ Manacher~et~al.~\cite{manacher1990optimum} proved that for
each $j=0,1,\dots, k$ it holds that $P_{j}$ is an island with respect to $C(P)$
and if $I\setminus P \neq \emptyset$ also $I \setminus P$ is an island with
respect to $C(P)$. We call such a decomposition of $P$ a
\emph{decomposition into covers and islands}. 

Manacher~et~al.~\cite{manacher1990optimum} also observed the following important
properties of a decomposition into covers and islands.

\begin{proposition}\label{proposition:covislanddecomp-props}
    Let $P = \greedyp(I) = (i_{1}, i_{2}, \dots, i_{l})$.
    \begin{enumerate}
        \item If $i_{j} \in C(P)$ it holds that $\iend(i_{j}) >
            \iend(i_{j+1})$.
        \item If $P = (P_{0}, c_{1}, P_{1}, c_{2}, \dots, c_{k}, P_{k})$ is a
            decomposition into covers and islands it holds that $L(P_{j}) =
            \emptyset$ for each $j=0,1,\dots,k$.
    \end{enumerate}
\end{proposition}

To illustrate the notions introduced above, consider the intervals in
Figure~\ref{figure:GP} given as a right endpoint-sorted list $I = (i_1,
i_2,\ldots, i_{12})$.  
\begin{figure}[h]
	\centering
    \begin{tikzpicture}[xscale=0.4,yscale=0.65]
        \draw (1,0) --  (4,0) node[align=left, above]{$i_{1}$\qquad};
        \draw (3,3) --  (10,3) node[align=left, above]{$i_{2}$\qquad};
        \draw (7,0) --  (11,0) node[align=left, above]{$i_{3}$\qquad};
        \draw (2,2) --  (12,2) node[align=left, above]{$i_{4}$\qquad};
        \draw (14,0) --  (16,0) node[align=left, above]{$i_{5}$\qquad};
        \draw (0,1) --  (19,1) node[align=left, above]{$i_{6}$\qquad};
        \draw (21,1) --  (25,1) node[align=left, above]{$i_{7}$\qquad};
        \draw (20.5,0) --  (27,0) node[align=left, above]{$i_{8}$\qquad};
        \draw (26.2,1) --  (29,1) node[align=left, above]{$i_{9}$\qquad};
        \draw (18,3) --  (30,3) node[align=left, above]{$i_{10}$\qquad};
        \draw (30.7,1) --  (32.5,1) node[align=left, above]{$i_{11}$\qquad};
        \draw (23,2) --  (33,2) node[align=left, above]{$i_{12}$\qquad};
    \end{tikzpicture}
	\caption{An interval model $I$ of twelve endpoint-sorted intervals \cite{chang1999deferred}.}  
	\label{figure:GP} 
\end{figure}
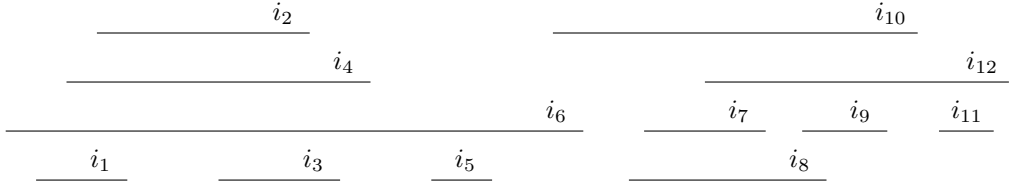
Algorithm $\greedyp$ starts by setting $P=(i_1)$. Neighbors of $i_1$ are $\{
i_2, i_4, i_6\}$, and since $\iend(i_2) < \min\{\iend(i_4), \iend(i_6)\}$ we
extend $P$ by $i_2$, i.e. $P = (i_1, i_2)$. Among neighbors of $i_2$ that are
not already in $P$, $i_3$ has the smallest right endpoint, so $P$ is extended to
$P=(i_1, i_2, i_3)$. Next candidates for the extension are $\{i_4, i_6 \}$ among
which we chose $i_4$, i.e. $P=(i_1, i_2, i_3, i_4)$. Next, the only possible
extension is by $i_6$, hence  $P=(i_1, i_2, i_3, i_4, i_6)$. Among the next candidates for extension $\{i_5, i_{10} \}$, interval $i_5$ is chosen. At this point the algorithm terminates and outputs $P = (i_1, i_2, i_3, i_4, i_6, i_5)$, since there is no neighbor of $i_5$ that is not already in $P$.

Now we find a decomposition into covers and islands of $P$. Since $\iend(i_6) >
\iend(\pend(P)= i_5)$, we have that $L(P) = \{i_6\}$, and  $C(P) = \{i_6\} \cup
C(P'=(i_1, i_2, i_3, i_4))$. $L(P')$ is the empty set, so the decomposition
process is over and we have that the decomposition into covers and islands of
$P$ is given by $C(P) = \{i_6\}$ and $P=(P_0, i_6, P_1)$, where $P_0= (i_1, i_2,
i_3, i_4)$ and $P_1= (i_5)$. Note that $P_0$, $P_1$ and $I \setminus P$ are
islands with respect to $C(P)= \{i_6\}$. Furthermore, note that our
decomposition satisfies the properties in
Proposition~\ref{proposition:covislanddecomp-props}.
\medskip

Given the fact that in the Steiner variant of the problem only the
intervals in $S$ have to be visited, we introduce \emph{neglectable intervals}.
Let $P$ be the current path at any point of the algorithm and $i'$ be the next
extension. We call $i'$ \emph{neglectable} with respect to $\pend(P)$, if $i'
\notin S$ and $\iend(i') <
\iend(\pend(P))$, i.e. $\pend(P) \in L((P, i'))$. We modify the algorithm
$\greedyp$, such that it skips neglectable
intervals with respect to the end of the current path. Analogously to $\greedyp$ this modification is denoted by
$\greedyp_{S}$. We define the set $N_{i}$ of intervals  that are not contained in
$\greedyp_{S}(I)$ since they are neglectable with respect to $i$ for some path $P$ during
the execution of $\greedyp_{S}$, where $i = \pend(P)$. We denote by $N$ the set of all
such neglectable intervals obtained during the entire run of $\greedyp_{S}$.

\begin{lemma}\label{lemma:scovisland-decomp}
    Let $P = \greedyp_{S}(I)$ be the path obtained by $\greedyp_{S}$ for a given list of
    intervals $I$ and $P = (P_{0}, c_{1}, P_{1}, c_{2}, \dots, P_{k-1}, c_{k}, P_{k})$ its
    decomposition into covers and islands in $G(I\setminus N)$. Let $C(P) = \{c_{1}, c_{2},
    \dots, c_{k}\}$, then it holds for all
    $j=0,1,\dots,k$ that $P_{j} \cap S \neq \emptyset$, i.e. $P_{j}$ is an
    $S$\nobreakdash-island with respect to $C(P)$ in $G(I\setminus N)$.
    It even holds that $P_{j} \cup N_{c_{j}}$ contains at least one $S$\nobreakdash-island with
    respect to $C(P)$ in $G(I)$.
\end{lemma}
\begin{proof}
    It is easy to see that this decomposition into covers and islands exists, 
    since if $P = \greedyp_{S}(I)$ it
    follows by construction that $P = \greedyp(I\setminus N)$.

    The fact that $P_{j}$ is an $S$\nobreakdash-island with respect to $C(P)$ in
    $G(I\setminus N)$ is a trivial consequence of of the decomposition into
    covers and islands.
    Since $c_{j}$ is used before every interval in $N_{c_{j}}$ we have that the
    left endpoint of every interval in $N_{c_{j}}$ is larger than the left
    endpoint of $c_{j}$. The right endpoints of each of those intervals is
    smaller than the right endpoint of $c_{j}$ by definition of neglected
    intervals. But this directly implies that $C(P)$ separates also $N_{c_{j}}$
    from the rest of $G(I)$, except for possibly $P_{j}$.
\end{proof}

Based on this we can obtain an easy procedure to solve the Steiner path cover problem
on interval graphs. We start with $\mathcal{P} = \emptyset$ and apply the algorithm $\greedyp_{S}$.
After termination let $P = \greedyp_{S}(I)$. We add $P$ to our partial solution
$\mathcal{P}$ and find the
smallest index $j$ such that $i_{j} \in S$ and $i_{j}$ is not in any path
currently contained in $\mathcal{P}$. Then we apply $\greedyp_{S}$ again to the list of intervals
$i_{j}, i_{j+1}, \dots, i_{n}$, until all intervals in $S$ are covered by one of
the paths in $\mathcal{P}$. The algorithm terminates with the Steiner path cover
$\mathcal{P}$ as its output.

\begin{theorem}\label{theorem:greedypathcover}
    The Steiner path cover obtained by iterated application of $\greedyp_{S}$ is
    optimal.
\end{theorem}
\begin{proof}
    Let $P_{1}, P_{2}, \dots, P_{l}$ be the paths obtained by the given
    algorithm and $C' = \bigcup_{j=1}^{l} C(P_{j})$ be the union of all the covers in the decomposition into
    covers and islands of each path. Then, by repeated application of
    Lemma~\ref{lemma:scovisland-decomp} we obtain that there are $l+|C'|$
    $S$\nobreakdash-islands with respect to $C'$ in $G(I)$. By
    Proposition~\ref{proposition:pathcover} we then know that $\pi_{S}(G(I)) =
    l$, so our solution is an optimal Steiner path cover.
\end{proof}

To illustrate our algorithm for the Steiner path cover problem we again consider the
example in Figure~\ref{figure:GP}. In the case when $S= I$, i.e., all intervals
need to be covered, our algorithm runs $\greedyp_S(I)$ which outputs $P'=(i_1,
i_2, i_3, i_6, i_5)$, and then it runs $\greedyp_S(I\setminus P')$ which outputs
$P''=(i_7, i_8, i_9, i_{10}, i_{12}, i_{11})$, and the algorithm terminates.
Therefore, for $S=I$ we have that $\pi_S(I)=2$.
Now lets say that $S=\{i_2, i_4, i_6, i_8, i_{10}, i_{12}\}$. $\greedyp_S(I)$
starts with the element of $S$ with the smallest right endpoint which is $i_2$.
Then it extends the path with $i_3$, $i_4$ and then $i_6$. After that, the algorithm
neglects $i_5$ since $\iend(i_5) < \iend(i_6)$ and $i_5 \notin S$. Next, the path is
extended by $i_{10}$, then $i_7$ is neglected, but $i_8$ is added to the path
(since $i_8 \in S$). Then the path is extended by $i_9$ and finally by $i_{12}$.
Interval $i_{11}$ is neglected. The output of the algorithm is the path $P=(i_2, i_3,
i_4, i_6, i_{10}, i_8, i_9, i_{12})$, so $\pi_S(I)= 1$. Note that the key
factor that allowed us to cover the set $S$ with only one path is the fact that we
could neglect $i_5$.
\medskip

By using the Deferred-queue approach by Chang~et~al.~\cite{chang1999deferred}
this algorithm can be implemented in $O(n)$ time.

\section{The Steiner Cycle Problem}\label{sec:algsteinercyc}

To solve the Steiner cycle problem we first run our algorithm for the Steiner
cover problem (see Section~\ref{sec:algpathcover}). If $\pi_{S} > 1$ we know that there cannot exist a Steiner cycle.
Otherwise, let $P = (i_{1}, i_{2}, \dots, i_{l})$ be the obtained Steiner path in
$G(I)$.

Based on $P$ we construct two paths $Q$ and $R$. We start by setting $R = (i_{1})$ and $Q =
(i_{2})$. Then, we iteratively process the intervals $i_{3}$ to $i_{n}$.
If in the step of processing interval $i_{j}$ we have that  $\pend(Q) = i_{j-1}$,
we consider the following two cases. If $i_{j} \cap \pend(R) \neq \emptyset$, we
extend $R$ by $i_{j}$, i.e. $R = (R, i_{j})$. Otherwise, we extend $Q$ by
$i_{j}$, i.e. $Q = (Q, i_{j})$. If on the other hand in this step we have that $\pend(R) = i_{j-1}$
we check symmetrically if $i_{j} \cap \pend(Q) \neq \emptyset$. If
this is the case we extend $Q$ by $i_{j}$ and if not we extend $R$ by $i_{j}$.

If in the end of this process $\pend(Q) = i_{q}$ and $\pend(R) = i_{l}$, or
vice versa, we try to connect $Q$ and $\rev(R)$ to a Steiner cycle. To
achieve this we check if $\pend(Q)$ and $\pend(R)$ are directly connected, i.e.
$\pend(Q) \cap \pend(R)\neq \emptyset$, or if
there is an interval $i'$ among the intervals $I' \subseteq I$, whose right endpoints
$\iend(i') > i_{j}$ for all $j=1,2,\dots,l$ such that both $\pend(Q) \cap i'
\neq \emptyset$ and $\pend(R) \cap i' \neq \emptyset$. In any of those two cases we can
connect $Q$ and $\rev(R)$ to a Steiner cycle.
Otherwise, the algorithm returns that no Steiner cycle exists.

\begin{theorem}\label{theorem:greedycycle}
    The given algorithm correctly decides the existence of a Steiner cycle in
    $G(I)$ and obtains such a cycle if possible.
\end{theorem}
\begin{proof}
    If the algorithm finds a Steiner cycle this is obviously true. Also, by
    correctness of the algorithm for the Steiner path cover
    (Theorem~\ref{theorem:greedypathcover}), if no Steiner path
    is found we correctly determine that no Steiner cycle can exist.

    Otherwise, let us assume that the algorithm did not find a Steiner cycle. Without loss of generality, let $\pend(R) = i_{h}$ with $h <  l-1$ and
    consider the path $P' = (i_{1}, i_{2}, \dots, i_{h})$ and its decomposition
    into covers and islands.
    Since $R$ was not extended by any of the intervals $i_{h+2},
    i_{h+3} \dots, i_{n}$, we have that $C(P') \cup \{i_{h+1}\}$ separates the islands of $P'$ from
    $\{i_{h+2}, i_{h+3}, \dots, i_{l}\}$. In addition since $\pend(R)$ and
    $\pend(Q)$ could not be connected with any interval in $I'$ it holds for all
    interval $i' \in I'$ that $\istart(i') > \iend(i_{h})$. Combining this with
    point 2 of Proposition~\ref{proposition:covislanddecomp-props} we observe
    that $\{i_{h+2}, i_{h+3}, \dots, i_{l}\} \cup I'$ is non-empty and
    an $S$\nobreakdash-island with respect to $C(P') \cup \{i_{h+1}\}$.
    
    By Lemma~\ref{lemma:scovisland-decomp} there are at least $|C(P')|+1$
    $S$\nobreakdash-island with respect to $C \cup \{i_{h+1}\}$. So, by
    Proposition~\ref{proposition:steinercycle} there does not exist a Steiner
    cycle in $G(I)$.
\end{proof}

Given a Steiner path $P$, the paths $Q$ and $R$ can be easily constructed in
$O(n)$ time. This gives a linear time algorithm for the Steiner cycle problem in
interval graphs.
\medskip

Now we illustrate our algorithm for the Steiner cycle problem on interval graphs
with the example given in Figure~\ref{figure:cyclealg}.
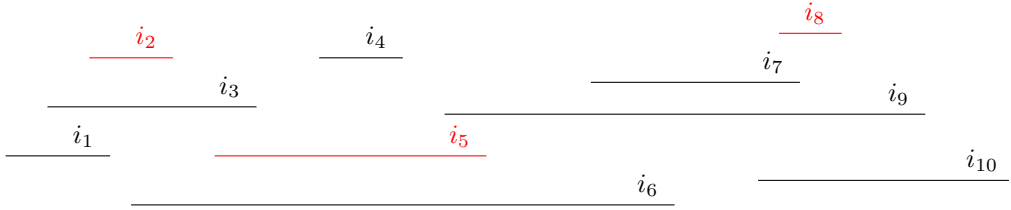
\begin{figure}[h]
	\centering
    \begin{tikzpicture}[xscale=0.55,yscale=0.65]
        \draw (1,1) --  (3.5,1) node[align=left, above]{$i_{1}$\qquad};
        \draw[red] (3,3) --  (5,3) node[align=left, above]{$i_{2}$\qquad};
        \draw (2,2) --  (7,2) node[align=left, above]{$i_{3}$\qquad};
        \draw (8.5,3) --  (10.5,3) node[align=left, above]{$i_{4}$\qquad};
        \draw[red] (6,1) --  (12.5,1) node[align=left, above]{$i_{5}$\qquad};
        \draw (4,0) --  (17,0) node[align=left, above]{$i_{6}$\qquad};
        \draw (15,2.5) --  (20,2.5) node[align=left, above]{$i_{7}$\qquad};
        \draw[red] (19.5,3.5) --  (21,3.5) node[align=left, above]{$i_{8}$\qquad};
        \draw (11.5,1.85) --  (23,1.85) node[align=left, above]{$i_{9}$\qquad};
        \draw (19,0.5) --  (25,0.5) node[align=left, above]{$i_{10}$\qquad};
    \end{tikzpicture}
	\caption{An instance of the Steiner cycle problem on an interval graph with $S=\{i_2, i_5, i_8\}$.}  
	\label{figure:cyclealg} 
\end{figure}
The given instance has 10 intervals $I=\{i_1, i_2, \ldots, i_{10}\}$ and $S=\{i_2, i_5, i_8\}$. 
Intervals in $S$ are represented with the red color. First we run
$\greedyp_S(I)$. It starts the path with $i_2$ and then extends it with $i_3$
and $i_5$ before neglecting $i_4$. Then it proceeds by extending the path with
$i_6$, $i_7$, finishing with $i_8$. Hence it obtains the Steiner path $P = (i_2,
i_3, i_5, i_6, i_7, i_8)$. In an attempt to create a Steiner cycle, we partition
$P$ into two paths $R$ and $Q$. We initialize them with the first two intervals in
$P$, that is, $R=(i_2)$ and $Q=(i_3)$. Now we consider $Q$ to be the current
path, and $R$ to be the previous path. In each step we consider the next
interval of $P$, and in the case that it intersect the end of the previous path,
we extend the previous path and make it the current path.
Otherwise we add the interval to the current path. So, interval $i_5$ is the
next interval in $P$, and it does not intersect $\pend(R)= i_2$, hence we add it
to $Q$, making it $Q=(i_3, i_5)$. The next interval is $i_6$, and it intersects
$\pend(R)= i_2$, hence we extend $R$ and make it the current path, so
$R=(i_2, i_6)$. Next interval $i_7$ does not intersect $\pend(Q)= i_5$ so we
extend $R$ again, making it $R=(i_2,i_6, i_7)$. Finally, interval $i_8$ does not
intersect $\pend(Q)= i_5$ so we extend $R$, making it $R=(i_2, i_6, i_7, i_8)$.
This ends our partition of $P$ with the resulting subpaths $R=(i_2, i_6, i_7,
i_8)$ and $Q=(i_3, i_5)$. Since $\pend(R)=i_8$ and $\pend(Q)=i_5$ do not
intersect, we cannot connect them into a cycle. The only remaining chance to do so is using an
interval from $I' = \{i \in I\setminus P\colon \iend(i) > \iend(\pend(P))\}=
\{i_9, i_{10}\}$. Luckily, $i_9$ intersect both $\pend(R)=i_8$ and
$\pend(Q)=i_5$, and can be used to connect $R$ and $Q$ into a cycle. The Steiner
cycle is then given by $(R, i_9, \rev(Q)) = (i_2, i_6, i_7, i_8,
i_9, i_5, i_3)$.

Now let us consider a modified instance of Figure~\ref{figure:cyclealg}, where
$i_4$ is also an element of $S$. Then  $\greedyp_S(I)$ would output the path
$P=(i_2, i_3, i_5, i_4, i_6, i_7, i_8)$, and the subsequent partition of $P$
would give $R=(i_2, i_6, i_7, i_8)$ and $Q=(i_3, i_5, i_4)$. But now there is no
interval in $I'$ that connects $\pend(R)=i_8$ and $\pend(Q)=i_4$, so our
algorithm outputs that there is no Steiner cycle. In order to verify that
there is no Steiner cycle we can follow the arguments in the proof of
Theorem~\ref{theorem:greedycycle}, which gives us a cutset $C=\{i_5, i_6\}$ that
separates $I$ into three  $S$-islands, and hence, by
Proposition~\ref{proposition:steinercycle}, guarantees that there is no Steiner
cycle.

\section{Streaming Algorithms -- The Problem of Limited Screen
Size}\label{sec:streaming}

An important question when considering solvability of game levels is which parts
of a level have to be visible to the user at any time for them to
deterministically know how to play correctly. To answer this question for our toy
model, we study the algorithms from Section~\ref{sec:algpathcover} and
\ref{sec:algsteinercyc} as streaming algorithms. We assume that the input stream
is presented as a sequence of right endpoint sorted intervals which can only be
examined in one pass. As its output the streaming algorithm has to write the list
of intervals giving the paths or cycle.

First, consider the algorithm $\greedyp_{S}$. In each step this algorithm needs access to the
next interval on the stream that is connected with the current path $\pend(P)$. If the next
interval $i$ on the stream is not connected to $\pend(P)$ there can be two reasons.
This interval could either be in a new different connected component than $P$,
or it could be connected to $P$ via another interval $i'$ with $\iend(i') > \iend(i)$.
Intervals of this kind are all completely contained in $i'$. After processing
and storing all
such intervals we clearly know whether the graph is disconnected or the path $P$ can
be extended and we can further process the stored intervals.
This motivates the introduction of the parameter $\kappa(I)$,
the maximum number of intervals contained in another. Based on this parameter we observe
that $\greedyp_{S}$ can be implemented as a single pass streaming algorithm with
$O(\kappa(I))$ additional storage. Based on this we obtain the following result.

\begin{theorem}\label{thm:streamingcover}
    Given $\kappa(I)$ the Steiner path cover problem on interval graphs can be solved
    by a single pass streaming algorithm in $O(n)$ time with $O(\kappa(I))$
    additional storage.
\end{theorem}
\begin{remark}
    If $\kappa(I)$ is not known to the algorithm the same result only holds assuming $G(I)$ is
    connected. Otherwise in the case of a disconnected interval graph the
    algorithm can not decide after $O(\kappa(I))$ steps that the graph is
    disconnected. It has to continue to store the intervals from the
    stream till the end, because there is no way of knowing if a future interval will be
    connected to $\pend(P)$ for the current path $P$.

    On the other hand if $\kappa(I)$ is known we can stop this process after
    storing $\kappa(I)$ intervals since we know that no more of them can be
    contained in another interval and terminate with the current path $P$.
\end{remark}

To solve the Steiner cycle problem, a single pass streaming algorithm can no longer first run
$\greedyp_{S}$ and then construct the two paths $Q$ and $R$, since this would need
two passes. Also the output of the cycle is only possible in a single pass,
without a large amount of additional memory, if the two paths $Q$ and $R$ are
accepted as an output instead of the list for the Steiner cycle.
In the application to platform games this is not a
problem since here a player actually is doing first a pass from the left to the
right and then another pass from the right to the left. So correct construction
of $Q$ during the first pass is enough to guarantee the possibility of
getting back to the exit later. This path for the way back can then be easily found
doing a simple greedy approach (see the description in the end of the current
section).

The construction of $Q$ and $R$ can be incorporated into the streaming variant
of $\greedyp_{S}$ described above without the need for additional memory. In
addition to $\pend(P)$ we also store $\pend(Q)$ and $\pend(R)$. This way 
in each step of the algorithm we can decide whether the next interval extending $P$
should be appended to $Q$ or $R$, by the same method as explained in
Section~\ref{sec:algsteinercyc}. This only needs additional memory for
storing both $\pend(Q)$ and $\pend(R)$ compared to just executing
$\greedyp_{S}$.

\begin{theorem}
    Given $\kappa(I)$ the Steiner cycle problem on interval graphs can be solved
    by a single pass streaming algorithm in $O(n)$ time with $O(\kappa(I))$
    additional storage.
\end{theorem}

It is important to note that from the view of a player the additional storage in
the streaming algorithms does not correspond to storage needed to decide the
next step of the game but to the range of the level that has to be visible to
the player. It covers the fact that the player has to be able to see at least
the next two intervals reachable from its current position and all the intervals
before that in a right endpoint sorted order.
The two things a player needs to remember at each point of the game are
 $\pend(Q)$ the platform it is currently on and $\pend(R)$. The algorithm can
 also be simplified in the following way.
 
 Assume the player is currently located on the interval $\pend(Q)$.
 There are two possible cases. In the first case the last step was
 jumping onto $\pend(Q)$. Let $i$ be the interval reachable from
 $\pend(Q)$ with $\iend(i)$ minimum, such that $i$ is not neglectable with
 respect to $\pend(Q)$. If $i \cap \pend(R) \neq \emptyset$ we extend $R$, so
 the player remembers $\pend(R) = i$. Otherwise the player jumps to $i$, so $\pend(Q) = i$.
 If neither is possible the current level is unsolvable.
 In the second case the last step was an extension of $R$, so $\pend(R)$ was
 updated. Let $i$ be the interval reachable from $\pend(R)$ with $\iend(i)$
 minimum, such that $i$ is not neglectable with respect to $\pend(R)$. If $i \cap
 \pend(Q) \neq \emptyset$ the player jumps to $i$, so $\pend(Q) = i$. Otherwise
 we extend $R$ so the player remembers that $\pend(R) = i$. If neither is
 possible the current level is also unsolvable.
 If the last interval in $S$ is either visited by the player, i.e. is equal to
 $\pend(Q)$ or reached by $R$, i.e. is equal to $\pend(R)$ the player tries to
 reach $\pend(R)$ from $\pend(Q)$ by jumping there directly or using an interval
 $i'$ with $\iend(i') > \max\{\iend(\pend(Q)), \iend(\pend(R)\}$. If this is not
 possible the player determines that the current level is unsolvable. Otherwise
 it can easily get back to the exit visiting all the unvisited intervals in $S$
 by reconstructing a maybe permuted version of the path $\rev(R)$. Let $i$ be
 the interval the player is currently on. In each step
 it can greedily jump to the reachable interval $i'$ with maximum left endpoint
 $\istart(i')$, such that $i'$ is not neglectable with respect to $i$ in the
 reverse sense. This means we can neglect jumping to $i'$ if $\istart(i) <
 \istart(i')$ and $i' \notin S$. This is just an application of $\greedyp_{S}$
 in reverse direction. Since the path $R$ exists, by the optimality of
 $\greedyp_{S}$ for the path cover problem, using this strategy the player
 finds a path $R'$ covering all intervals in $S$ and returning to the start of
 the level.

\section{Conclusion}

We obtained linear time algorithms for both the Steiner path cover problem and
the Steiner cycle problem, assuming the intervals are given as a right endpoint
sorted list. We also analyzed those algorithms as single pass
streaming algorithms to study solvability of a simplified model for platform
game levels.

Our simplification reduced those levels to a one-dimensional
interval graph model. The hamiltonian cycle and path problems for
two-dimensional generalizations of interval graphs are known to be NP-hard. It
would be of interest to study special cases of these problems inspired from game
levels. Furthermore the analysis of streaming algorithms for interval graphs is
a natural extension to classic algorithms for interval graphs. Understanding
other efficient algorithms for different problems on interval graphs in this
model is a very interesting area for further research.



\bibliography{steiner-interval}

\end{document}